\documentclass[11pt,letterpaper]{article}

\usepackage{amsmath}
\usepackage{amssymb}
\usepackage{amsthm}

\usepackage{graphicx}
\usepackage[latin1]{inputenc}
\usepackage{fullpage}
\usepackage{color}
\usepackage{bbm}
\usepackage{xspace}
\usepackage{titlesec}
\usepackage{times}
\usepackage{paralist}
\usepackage{bm}
\usepackage{stmaryrd}
\usepackage{nicefrac}
\usepackage{ifpdf}

\definecolor{Darkblue}{rgb}{0,0,0.4}
\definecolor{Brown}{cmyk}{0,0.81,1.,0.60}
\definecolor{Red}{rgb}{1,0,0}
\newcommand{\mydriver}{hypertex}
\ifpdf
 \renewcommand{\mydriver}{pdftex}
\fi
\usepackage[breaklinks,\mydriver]{hyperref}
\hypersetup{colorlinks=true,
            citebordercolor={.6 .6 .6},linkbordercolor={.6 .6 .6},%
citecolor=Darkblue,urlcolor=black,linkcolor=Red,pagecolor=black}
\newcommand{\lref}[2][]{\hyperref[#2]{#1~\ref*{#2}}}

\newtheorem{lemma}{Lemma}
\newtheorem{thm}{Theorem}
\newtheorem{claim}{Claim}
\newtheorem{cor}{Corollary}
\newtheorem{prop}{Proposition}

\newtheorem{definition}{Definition}

\renewcommand{\d}{\textrm{d}}

\newcommand{\ip}[2]{\langle #1, #2\rangle}
\newcommand{\R}{\mathbb{R}}

\newcommand{\E}{\mathbb{E}}

\newcommand{\OPT}{\textsc{OPT}}
\newcommand{\GAP}{\textsc{Generalized Assignment Problem}\xspace}
\newcommand{\GAPs}{\textsc{GAP}\xspace}
\newcommand{\stocp}{\textsc{Stochastic $\ell_p$ Load Balancing Problem}\xspace}
\newcommand{\stocps}{\textsc{StochLoadBal$_p$}\xspace}

\renewcommand{\P}{\mathcal{P}}

\newcommand{\Var}{\mathrm{Var}}

\newcommand{\e}{\varepsilon}

\newcommand{\J}{Y}
\newcommand{\tJ}{Y'}
\newcommand{\rJ}{Y''}
\newcommand{\tildeJ}{\tilde{Y}'}

\newcommand{\one}{\bm{1}}

\newcounter{mynotes}
\newcommand{\pSel}{\textsc{SubsetSelection$_p$}\xspace}

\newcommand{\vertiii}[1]{{\left\vert\kern-0.25ex\left\vert\kern-0.25ex\left\vert #1 \right\vert\kern-0.25ex\right\vert\kern-0.25ex\right\vert}}
\newcommand{\lf}{\nu}
\newcommand{\lfT}{\nu^+}

\title{Stochastic $\ell_p$ Load Balancing and Moment Problems \\via the $L$-Function Method}
\author{Marco Molinaro}
\date{}

\begin{document}
	\maketitle
	
	\begin{abstract}
		This paper considers stochastic optimization problems whose objective functions involve powers of random variables. For a concrete example, consider the classic \stocp (\stocps): There are $m$ machines and $n$ jobs, and we are given independent random variables $\J_{ij}$ describing the distribution of the load incurred on machine $i$ if we assign job $j$ to it. The goal is to assign each job to the machines in order to minimize the expected $\ell_p$-norm of the total load incurred over the machines. That is, letting $J_i$ denote the jobs assigned to machine $i$, we want to minimize $\E (\sum_i(\sum_{j \in J_i} Y_{ij})^p)^{1/p}$. While convex relaxations represent one of the most powerful algorithmic tools, in problems such as \stocps the main difficulty is to capture such objective function in a way that only depends on each random variable separately.
		
		In this paper,  show how to capture $p$-power-type objectives in such separable way by using the \emph{$L$-function} method. This method was precisely introduced by Lata{\l}a to capture in a sharp way the moment of sums of random variables through the individual marginals. We first show how this quickly leads to a constant-factor approximation for very general subset selection problem with $p$-moment objective.
		
		Moreover, we give a constant-factor approximation for \stocps, improving on the recent $O(\frac{p}{\ln p})$-approximation of [Gupta et al., SODA 18]. Here the application of the method is much more involved. In particular, we need to prove structural results connecting the expected $\ell_p$-norm of a random vector with the $p$-moments of its coordinate-marginals (machine loads) in a sharp way, taking into account simultaneously the different scales of the loads that are incurred in the different machines by an \emph{unknown} assignment. Moreover, our starting convex (indeed linear) relaxation has exponentially many constraints that are not conducive to integral rounding; we need to use the solution of this LP to obtain a reduced LP which can then be used to obtain the desired assignment. 	
	\end{abstract}
	
\thispagestyle{empty}
\newpage

\setcounter{page}{1}


\section{Introduction} \label{sec:intro}

		This paper considers stochastic optimization problems whose objective functions are related to powers of sums of random variables. For a concrete example, consider the classic \stocp (\stocps): There are $m$ machines and $n$ jobs, and we are given independent non-negative random variables $\J_{ij}$ (job sizes) describing the distribution of the load incurred on machine $i$ if we assign job $j$ to it. The goal is to, only knowing these distributions, assign each job to the machines in order to minimize the expected $\ell_p$-norm of the realized total load incurred over the machines. That is, letting $J_i$ denote the jobs assigned to machine $i$, we want to minimize
		\vspace{-3pt}
		\begin{align}
		\E \bigg\|\bigg(\sum_{j \in J_i} Y_{ij}\bigg)_{i \in [m]}\bigg\|_p = \E \Bigg[\sum_{i \in [m]} \bigg(\sum_{j \in J_i} Y_{ij}\bigg)^p\Bigg]^{1/p}, \label{eq:obj}
		\end{align}
		if $p \in [1,  \infty)$, and to minimize the makespan $\E \|(\sum_{j \in J_i} Y_{ij})_{i \in [m]}\|_{\infty} = \E[\max_i \sum_{j \in J_i} Y_{ij}]$ when $p = \infty$. Notice the entire assignment is done up-front without knowledge of the actual outcomes of the random variables, and hence there is no adaptivity. We remark that the $\ell_p$-norms interpolate between considering only the most loaded machine ($\ell_\infty$) and simply adding the loads of the machines ($\ell_1$), and have been used in this context since at least the 70's~\cite{chandra,cody}, since in some applications they better capture how well-balanced an allocation is~\cite{awerbuch}. This classic problem has been widely studied in its stochastic~\cite{KRT,goelIndyk,anupamSLB,pinedo}, deterministic~\cite{Lenstra1990,allNorms,AzarConvexLB,srinivasanGeneralSched,sviridenkoDecoup}, and online versions~\cite{awerbuch,uniformMachines,charikarLoadBal,caragiannis,CaragiannisRestricted,molinaro}. See \cite{anupamSLB} for a comprehensive discussion and literature review on \emph{stochastic} load balancing, most relevant for us. 
				
		The deterministic versions of such problems can typically be well-approximated through the use of convex programs; for example, this method has provided constant-factor approximations for the deterministic version of $\stocps$~\cite{AzarConvexLB,srinivasanGeneralSched,sviridenkoDecoup}. However, in the stochastic version of these problems the situation is much more complicated, since in principle terms like \eqref{eq:obj} require multi-dimensional integration due to the expectation involving powers of sums of random variables.
		
		Thus, the main element for using convex programs to tackle such stochastic problems is to be able to approximately capture the objective function in a way that only depends on each random variable individually. The first idea is to replace the random variables by just their expectation, for example reducing \eqref{eq:obj} to $\|(\sum_{j \in J_i} \E Y_{ij})_{i \in [m]}\|_p$. Unfortunately, even basic examples show that too much is lost and this simple proxy is not enough.
		 For the special case of \stocps with $p=\infty$ and identical machines (i.e., the item sizes are independent of the machines),~\cite{KRT} proposed to use the so-called \emph{effective size}~\cite{hui} of a job as a proxy instead of its expectation: For a random variable $X$ and parameter $\ell \in (1,\infty)$, its \emph{effective size} (at scale $\ell$) is
		\begin{align}
			\beta_\ell(X) := \frac{1}{\ln \ell} \cdot \ln \E\left(e^{X \cdot \ln \ell} \right); \label{eq:effectiveSize}
		\end{align}
		for $\ell = 1$, it is defined $\beta_1(X) := \E X$. Using this notion,~\cite{KRT} obtained the first constant-factor approximation for this special case of \stocps. They also use it to provide approximations for stochastic bin-packing and knapsack problems (all packing- or $\ell_\infty$-type problems). Only recently, Gupta et al.~\cite{anupamSLB} managed to use this  fruitful notion to obtain a constant approximation for the unrelated machines case (but still $p = \infty$). 
		
		However, suitable notions of effective size have not been used for $p$-power-type functions. For example, for \stocps with general $p$ only an $O(\frac{p}{\ln p})$-approximation is known, also due to~\cite{anupamSLB}, and relies on other techniques (expected size as a proxy plus Rosenthal's Inequality). Oddly, this approximation ratio goes to infinity as $p \rightarrow \infty$, despite the constant-factor approximation known for $p = \infty$. This indicates our current shortcomings in algorithmic and analytical tools for dealing with such moment-type objectives.   



	\subsection{Our results and techniques}
		
		In this paper we show how to approximate stochastic optimization problems with moment or $p$-power-type objectives using the \emph{$L$-function}\footnote{The name $L$-function is borrowed from~\cite{deLaPena}.} method. This method was precisely introduced by Lata{\l}a~\cite{latala} to capture in a sharp way the moment of sums of random variables by only looking at each of them separately. Using this method is quite simple and we hope it will find many additional applications.
		
	\paragraph{$L$-function method.} Consider non-negative independent random variables $X_1,\ldots,X_n$, and suppose we want to better understand the (raw) $p$-moment of their sum: $(\E (\sum_i X_i)^p)^{1/p}$. For that,~\cite{latala} defines a notion of ``effective size'' $\nu_{\e,p}(X_i)$ that depends on the power $p$ and on the additional parameter $\e$ with the following property (see Theorem~\ref{thm:lFunc}):
		\begin{align}
		\sum_i \lf_{p,\e}(X_i) = 1 ~~\Leftrightarrow~~ \e \approx \bigg(\E \bigg(\sum_i X_i\bigg)^p\bigg)^{1/p}, \label{eq:introLfunc}
		\end{align} 
		where in the approximation only constant factors are lost. This result shows that the moment of sums of random variables does not depend much on a \emph{stochastic interaction} between them, only on the interaction of the \emph{deterministic proxies} $\lf_{p,\e}(X_i)$. One difficulty is that this deterministic interaction has an \emph{implicit} form that depends on setting the parameter $\e$ in the ``right'' way.

		\paragraph{Quick application: Subset selection with $p$-moment objective.} Nonetheless, to show how the $L$-function method yields in a simple way approximations in the context of optimization problems, we consider the following general subset selection problem with moment objective (\pSel): There are $n$ items, the value of item $j$ is stochastic and given by the non-negative random variable $V_j$, and these random variables are independent. Given a subset $\P \subseteq \{0,1\}^n$ of the boolean cube representing the feasible sets of items, the goal is find a feasible set that maximizes the $p$-moment of the sum of the selected items' values:
	\begin{align*}
		\max_{x \in \P}~\bigg(\E \bigg(\sum_j V_j x_j\bigg)^p\bigg)^{1/p}.
	\end{align*}
	
	Using the $L$-function method, we show that one can reduce this problem to that of optimizing a \emph{deterministic linear function} over $\P$. 
	
	\begin{thm} \label{thm:pSel}
		Suppose there is a constant approximation for optimizing any non-negative linear function over $\P$ (i.e., for any non-negative vector $c \in \R^n_+$, we can find a point $\bar{x} \in \P$ satisfying $\ip{c}{\bar{x}} \ge \Omega(1)\cdot\max_{x \in \P} \ip{c}{x}$). Then there is a constant approximation for \pSel over $\P$ for any $p \in [1,\infty)$.
	\end{thm}
	
	The proof is very simple: By standard binary search arguments, we can assume we know the optimal objective value $\OPT$. Then based on equation \eqref{eq:introLfunc}, set $\e = \OPT$. The ``$\Rightarrow$'' direction of this equation essentially shows that to get value $\approx \OPT$ it suffices to find a solution $x \in \P$ with $\sum_j \lf_{p,\e}(V_j) x_j \ge 1$ (a deterministic linear feasibility/optimization problem), and the direction ``$\Leftarrow$'' essentially shows that the optimal solution satisfies this inequality, thus such solution can indeed be found. We carry this out more formally in Appendix \ref{app:pSel}.		
	

		\paragraph{Another application: \stocps.} In our next and main result, we show that \stocps admits a constant approximation for all $p \in [1,\infty]$, improving over the $O(\frac{p}{\ln p})$-approximation of~\cite{anupamSLB}. 

	\begin{thm} \label{thm:LB}
		For all $p \in [1,\infty]$, the \stocp admits a constant-factor approximation. 
	\end{thm}

	In this case the application of the $L$-function method is much more involved. The first issue is that the objective function \eqref{eq:obj} is not of the form that can be tackled directly by the $L$-function method. To connect the two, we prove a bound relating the expected $\ell_p$-norm of the sum of random variables and the $p$-moments of these sums. One direction is easy: given independent RVs $\{X_{ij}\}_{ij}$ and letting $S_i = \sum_j X_{ij}$, by the concavity of $x \mapsto x^{1/p}$ Jensen's inequality gives  
	\begin{align}
		\E \|(S_1,\ldots,S_m)\|_p \le \bigg(\E \|(S_1,\ldots,S_m)\|_p^p  \bigg)^{1/p} = \bigg(\sum_i \E S_i^p \bigg)^{1/p},  \label{eq:jensen}
	\end{align}
	so the expected $\ell_p$-norm is upper bounded by the moments $\E S_i^p$.	However, the other direction (with constant factor loss) is not true in general. Nonetheless, we prove such converse inequality under additional assumptions on the moments $\E S_i^p$ (that are discharged later); this is done in Section \ref{sec:lpMoments}. 
	
	Given this result, the idea is to write an assignment LP with additional linear constraints based on the $\lf_{p,\e_i}(\J_{ij})$'s to control the moment of the loads in each of the machines, and thus the objective function \eqref{eq:obj}. But the second issue appears: even if we assume to know the optimal objective value $\OPT$, we do not know the  moment of the loads in each machine in the optimal solution, needed to set the parameters $\e_i$. Thus, we need to write a \emph{valid constraint for each of the possible combination of $\e_i$'s}. The general theory behind it is developed in Section \ref{sec:lFuncLp}, and the LP is presented in Section \ref{sec:startingLP}. Addressing a similar issue in the case $p = \infty$ was a main contribution of~\cite{anupamSLB} and we borrow ideas from it, though in the case $p < \infty$ they need to be modified to avoid super-constant losses, see discussion in Section \ref{sec:lFuncLp}.
	
	Finally, as indicated, this LP has a large (exponential in $m$) number of inequalities, and thus it seems unlikely one can convert a fractional solution into an integral assignment satisfying all of the constraints. Thus, again inspired by~\cite{anupamSLB}, we use the optimal solution of this LP to obtain an estimate of the ``right'' $\e_i$'s for each of the machines and write a reduced LP based on them. This reduced LP is essentially one for the \GAP, for which one can use the classic algorithm by Shmoys and Tardos~\cite{ST} to obtain an approximate integral assignment. 
	
	We remark that even in the deterministic version of the problem previous approximations relied on convex programs~\cite{AzarConvexLB,srinivasanGeneralSched,sviridenkoDecoup}, so our techniques also give the first LP-based approach in this case.

	
\subsection{Notation}

	Unless specified, the letter $p$ always denotes a value in $(1,\infty]$. Given a vector $v \in \R^m$, its $\ell_p$-norm is defined by $\|v\|_p := (\sum_i v_i^p)^{1/p}$. Given a subset of coordinates $I \subseteq [m]$, we use $v_I = (v_i)_{i \in I}$ to denote the restriction of $v$ to these coordinates. When computationally relevant, we assume that the input distributions are discrete, supported on a finite set, and given explicitly, i.e., for each $x$ in the support we are given $\Pr(X = x)$.


\section{The L-function method} \label{sec:LFunc}

	\begin{definition}
		For any random variable $X$ and parameters $p,\e > 0$, functional $\lf_{\e,p}$ is defined as $$\lf_{\e,p}(X) := \frac{1}{p}\, \ln\bigg(\E\bigg|1 + \frac{X}{\e}\bigg|^p \bigg).$$
	\end{definition}	
		
		To simplify the notation, we omit the subscript $p$ in $\lf_{\e,p}$. As mentioned above, the main property of this functional is the following:

		\begin{thm}[Theorem 1.5.2 of~\cite{deLaPena}] \label{thm:lFunc}
			Consider non-negative independent random variables $X_1, \ldots, X_n$, and let $S = \sum_j X_j$. Let $\e^*$ be such that $\sum_j \lf_{\e^*}(X_j) = 1$. Then $$\left(\frac{\e^*}{10}\right)^p \le \E S^p \le (e \e^*)^p.$$
		\end{thm}
		
	It will be convenient to have slightly more flexible versions of these bounds. 
		
	\begin{lemma} \label{lemma:UBF}
		Let $X_1,\ldots,X_n$ and $S$ be as in Theorem \ref{thm:lFunc}. For any $\e > 0$, $$\E S^p \le \e^p e^{p \sum_j \lf_\e(X_j)}.$$
	\end{lemma}
	
	\begin{proof}
		This is the development in page 37 of~\cite{deLaPena}, which we reproduce for convenience. Using the inequality $1 + \sum_i a_i \le \prod_i (1 + a_i)$ valid for non-negative $a_i$'s, we have
		\begin{align*}
		 \E\bigg(\frac{1}{\e} \sum_j X_j \bigg)^p & \le  
		 \E\bigg(1 + \frac{1}{\e} \sum_j X_j \bigg)^p \le \E\bigg(\prod_j \bigg(1 + \frac{1}{\e} X_j\bigg) \bigg)^p \\
		 &\le \prod_j \E \bigg(1 + \frac{1}{\e} X_j\bigg)^p = e^{p \sum_j \lf_\e(X_j)}.
		\end{align*}
		Multiplying both sides by $\e^p$ concludes the proof. 
	\end{proof}
	
	\begin{lemma} \label{lemma:LBF}
		Let $X_1,\ldots,X_n$ and $S$ be as in Theorem \ref{thm:lFunc}. If $\sum_j \lf_\e(X_j) \ge 1$, then $\E S^p \ge \left(\frac{\e}{10}\right)^p$.
	\end{lemma}
	
	\begin{proof}
		Let $t$ be such that $\sum_j \lf_\e(X_j/t) = 1$; this is equivalent to $\sum_j \lf_{t \e}(X_j) = 1$. By our assumption and the fact $x \mapsto \lf_x$ is decreasing, notice that $t \ge 1$. Then from Theorem \ref{thm:lFunc} we have $\E S^p \ge (\frac{t\e}{10})^p \ge (\frac{\e}{10})^p$, which concludes the proof.  
	\end{proof}



\section{Towards Stochastic $\ell_p$ Load Balancing: Controlling $\ell_p$-norm in a separable way}

 Although we work on a more abstract setup, it may be helpful to think throughout this section that the $\{X_{i,j}\}_j$ represents the set of jobs assigned to machine $i$, and $S_i = \sum_j X_{i,j}$ represents the load of this machine.

\subsection{Relating expected $\ell_p$-norms and moments} \label{sec:lpMoments}

	The goal of this section is to relate the expected $\ell_p$-norm $\E \|(S_1,\ldots,S_m)\|_p$ of a random vector $S = (S_1,\ldots,S_m)$  and the coordinate moments $\E S_i^p$. As mentioned in the introduction, Jensen's inequality (inequality \eqref{eq:jensen}) gives the upper bound $\E \|S\|_p \le (\sum_i \E S_i^p )^{1/p}$; in this section we prove a partial converse to this inequality. 

	
	To see the difficulty in obtaining such converse suppose $S_1$ is a Poisson random variable with parameter $\lambda = 1$, and $S_2,\ldots,S_m = 0$. It is known that $\E S^p_1 \approx (\frac{p}{\ln p})^p$, and the Jensen's based inequality \eqref{eq:jensen} gives the upper bound $\E \|S\|_p \le (\sum_i \E S^p_i)^{1/p} \approx \frac{p}{\ln p}$. However, the actual expected norm is $\E \|S\|_p = \E S_1 = 1$. Thus, in general it is not possible to obtain a converse to the Jensen's based inequality without losing a factor of $\Omega(\frac{p}{\ln p})$.\footnote{This unavoidable gap of $\frac{p}{\ln p}$ between the moment $(\E S_1^p)^{1/p}$ and the expectation $\E S_1$ is one of the losses in the $O(\frac{p}{\ln p})$-approximation of \cite{anupamSLB}, which appears from the use of Rosenthal's Inequality.} Nonetheless, we show that one can obtain tighter bounds as long as none of the $S_i$'s contributes too much to the sum $\sum_i \E S_i^p$ (and each $S_i$ is a sum of ``small'' random variables). For that we need the following sharp moment comparison from~\cite{hitczenko2001}, which is a vast generalization of Khinchine's Inequality; we simplify the statement for our purposes, and for a RV $X$ use $\vertiii{X}_p := (\E X^p)^{1/p}$ to denote its $p$-th moment.  
	
	\begin{thm}[Theorem 6.2 of~\cite{hitczenko2001}] \label{thm:compMoments}
		Let $X_1, \ldots, X_n$ be independent real-valued random variables. Let $S = \sum_i X_i$, and $M = \max_i X_i$. Then there is a constant $c$ such that for all $p,q \ge 1$
		\begin{align*}
			\vertiii{S}_q \le c \cdot \frac{q}{\max\{p, \ln(e + q)\}} \bigg( \vertiii{S}_p + \vertiii{M}_q\bigg).
		\end{align*}
	\end{thm}
	
	Here is our converse to the Jensen's based inequality.
	
	\begin{lemma} \label{lemma:convJensen2}
		Let $\{X_{i,j}\}_{i,j}$ be independent random variables in $[0,1]$ and let $S_i := \sum_j X_{i,j}$. Suppose that $\E S_i^p \le 1$ for all $i$ but $(\sum_i \E S_i^p)^{1/p} \ge \frac{1}{\alpha}$ for a small enough constant $\alpha$ (independent of $p$). Then $\E \|S\|_p > \frac{1}{4} (\sum_i \E S_i^p)^{1/p}$.
	\end{lemma}
	
	\begin{proof}
		Since the random variables are non-negative, for all $t \ge 0$ we have
		\begin{align*}
		 \E \|S\|_p = \E \left(\sum_i S_i^p \right)^{1/p} \ge t^{1/p} \cdot \Pr\bigg(  \sum_i  S_i^p \ge t\bigg).
		\end{align*}
	It then suffices to show that with probability $\frac{1}{2}$ we have $\sum_i S^p_i \ge \frac{1}{2} \sum_i \E S_i^p$. For this, it suffices to upper bound by $\frac{1}{2}$ the probability that $\sum_i S^p_i \notin (1 \pm \frac{1}{2}) \sum_i \E S_i^p$. Let $\mu_i = \E S_i^p$ (note the exponent) and $\mu = \sum_i \mu_i$. Since the $S_i^p$'s are independent, we can apply Chebychev's Inequality to their sum to get 
		\begin{align}
			\Pr \left(\sum_i S_i^p \notin \left(1 \pm \frac{1}{2}\right) \mu\right) \le 4 \cdot \frac{\sum_i \Var(S_i^p)}{\mu^2} \label{eq:convJensen20}
		\end{align}
		
	To upper bound the variance of $S_i^p$, first note that $\Var(S_i^p) = \E S_i^{2p} - (\E S_i^p)^2 \le \E S_i^{2p}$. To upper bound the right-hand side, the idea is to use the moment comparison Theorem \ref{thm:compMoments} to obtain $\E S_i^{2p} \lesssim 2^p (\E S_i^p)^2 \le 2^p \E S_i^p$ (the last inequality by assumption).	More precisely, for each $i$ let $M_i = \max_j X_{i,j}$ denote the largest component of $S_i$ (in each scenario); applying Theorem \ref{thm:compMoments} with $q = 2p$ we have 
		\begin{align*}
			\E S_i^{2p} \le c^{2p} \cdot 2^{2p} \left[\left( \E S_i^p \right)^{1/p} + \left(\E M_i^{2p} \right)^{1/2p} \right]^{2p} \le c^{2p} \cdot 2^{4p} \left[\left( \E S_i^p \right)^2 + \E M_i^{2p} \right],
		\end{align*}
		where the last inequality uses $(a+b)^q \le (2\max\{a,b\})^q \le 2^q (a^q + b^q)$. Moreover, the assumption $\E S_i^p \le 1$ implies that $(\E S_i^p)^2 \le \E S_i^p$, and the assumption $X_{i,j} \in [0,1]$ implies $M_i^{2p} \le M_i^p \le S_i^p$. Therefore, we obtain that $\Var(S_i^p) \le 2c^{2p} \cdot 2^{4p} \cdot \E S_i^p$. Moreover, by assumption $\mu \ge \frac{1}{\alpha^p}$, and hence $\mu^2 \ge \frac{1}{\alpha^p}\,\sum_i \E S_i^p$. Employing these bound, we obtain that the right-hand side of \eqref{eq:convJensen20} is at most $8 \alpha^{p} \cdot c^{2p} \cdot 2^{4p}$. But for $\alpha$ a sufficiently small constant ($1/(c^2 2^8)$ suffices), this upper bound is at most $\frac{1}{2}$. This concludes the proof. 
	\end{proof}
	
	We will also need the following corollary, which is essentially Claim 3 of~\cite{anupamSLB} with a different parametrization; its proof is presented in Appendix \ref{app:simpleConv}.
	
	
	\begin{cor} \label{cor:simpleConv}
		Consider a scalar $\O$, let $\alpha$ be the constant in Theorem \ref{thm:compMoments}. Let $\{X_{i,j}\}_{i,j}$ be independent random variables in $[0,\alpha \O]$, and let $S_i := \sum_j X_{i,j}$. If $\E \|S\|_p \le \O$, then $\sum_{i,j} \E X_{i,j}^p \le (4\O)^p$.
	\end{cor}
	

	
	\subsection{Using the L-function method to control the $\ell_p$-norm} \label{sec:lFuncLp}
	
	\paragraph{A first attempt.}	Despite having the bound from the previous section it may still not be clear how we can use it to write an LP/IP that yields a good approximation for \stocps. We sketch a (failed) attempt of how we could try to proceed. Again consider independent RV's $\{X_{i,j}\}_{i,j}$ (e.g., assignment of jobs to machines) and let $S_i = \sum_j X_{i,j}$ (load of machine $i$). We claim that if $\E \|S\|_p \le \OPT$, then $\sum_i \sum_j \lf_{100 \OPT/m^{1/p}}(X_{i,j}) \le m$ (i.e., the optimal assignment satisfies these constraints). We informally sketch why this is the case under simplifying assumptions (let $\tilde{\lf}_\e(S_i) := \sum_j \lf_\e(X_{i,j})$):
	
	\begin{enumerate}
		\item By contradiction, suppose $\sum_i \tilde{\lf}_{100\OPT/m^{1/p}}(S_i) > m$, and \emph{assume} that for all $i$ we have \linebreak$\tilde{\lf}_{100\OPT/m^{1/p}}(S_i) \le O(1)$
		\item This implies that for $\Omega(m)$ indices $i$ we have $\tilde{\lf}_{100\OPT/m^{1/p}}(S_i) \ge 1$; recall this is around the ``right'' condition to apply the results from the L-function method
		\item More precisely, Lemma \ref{lemma:LBF} implies that for each such $i$ we have $\E S_i^p \ge (\frac{100\OPT}{10 m^{1/p}})^p = \frac{10^p \OPT^p}{m}$
		\item \emph{Assuming} the requirements of Lemma~\ref{lemma:convJensen2} are met, we can use it to obtain $\E \|S\|_p > \frac{1}{4} (\Omega(m) \frac{10^p \OPT^p}{m})^{1/p} > \OPT$ (the last inequality holds if we adjust the constants properly). This reaches the desired contradiction.  
	\end{enumerate}
	In fact, one can apply this argument to any subset $K \subseteq [m]$ of coordinates to obtain that
	\begin{align}
		\sum_{i \in K} \sum_j \lf_{100 \OPT/|K|^{1/p}}(X_{i,j}) \le |K|. \label{eq:everySubset}
	\end{align}
	Therefore, after guessing $\OPT$, we can write the following IP enforcing these restrictions and be assured that the optimal solution is feasible for it:
	\begin{align*}
		&\sum_{i \in K} \sum_{\ell}  \lf_{100\OPT/|K|^{1/p}}(\J_{i,\ell})\cdot x_{i,\ell} \le  |K|~~~~~\forall K \subseteq [m]\\
		&x \in \textrm{assignment polytope} \cap \{0,1\}^{m \times n}.
	\end{align*}
	In turn, suppose we can use this IP to obtain an integral assignment satisfying \eqref{eq:everySubset} (approximately). Then we can try to use the moment control from Lemma \ref{lemma:UBF} and the Jensen's-based inequality \eqref{eq:jensen} to reverse the process and argue that our solution has expected $\ell_p$-load $O(\OPT)$. For that, since the L-function method is most effective when $\e$ is such that $\tilde{\lf}_\e(S_i) \approx 1$, we can use the following idea inspired by \cite{anupamSLB}: assign to each machine $i$ a size $k$ such that $\tilde{\lf}_{\OPT/k^{1/p}}(S_i) \approx 1$ (or equivalently, $\E S_i^p \approx \frac{\OPT^p}{k}$); since by \eqref{eq:everySubset} there are at most $\approx k$ machines assigned to size $k$, we can hope to prove that $\sum_i \E S_i^P \lesssim \OPT^p$ and from the Jensen's based inequality obtain $\E \|S\|_p \lesssim \OPT$. Unfortunately this argument is not enough because the former inequality is not true: for $j = 1,\ldots, \ln m$ we could have $\frac{m}{2^j}$ machines with $\E S_i^p \approx \frac{\OPT^p}{m/2^j}$, thus assigned to size $\frac{m}{2^j}$, and get $\sum_i \E S_i^p \approx \sum_j \frac{m}{2^j} \frac{\OPT^p}{m/2^j} \approx (\ln m) \OPT^p$.

	\paragraph{Multi-scale bound.} The logarithmic loss in the previous example comes form the fact we grouped the machines with similar scale of moment $\E S_i^p$ and applied the upper bound \eqref{eq:everySubset} \emph{separately} for each group. To avoid this loss we will then obtain a more refined upper bound that takes into account all scales \emph{simultaneously}.

	\begin{thm} \label{thm:mainCstr}
	 Consider a scalar $\O$ and a sufficiently small constant $\alpha$. Consider independent random variables $\{X_{i,j}\}_{i,j}$ in $[0,\alpha \O]$. Let $S_i = \sum_j X_{i,j}$, and suppose $\E \|S\|_p \le \O$. Consider the scaled down variables $\tilde{X}_{i,j} := \frac{X_{i,j}}{44}$.
  Then for any sequence of values $v_1,\ldots,v_m \ge (1/\alpha)^p$, we have
		%
		\begin{align}
			\sum_i \frac{1}{v_i}\, \left(\sum_j \lf_{\O/v_i^{1/p}} \big(\tilde{X}_{i,j}\big) - 1\right) \le 3. \label{eq:mainCst}
		\end{align}
	\end{thm}
	
\noindent (For example, when $v_i = m$ for all $i$ this corresponds roughly to the bound \eqref{eq:everySubset} with $K = [m]$.) 
	
	\begin{proof}
 To simplify the notation let $\e_i := \frac{\O}{v_i^{1/p}}$ and define $\tilde{S}_i := \sum_j \tilde{X}_{i,j}$. The high-level idea is to show that if \eqref{eq:mainCst} does not hold then Lemmas \ref{lemma:LBF} and \ref{lemma:convJensen2} imply that $\E \|S\|_p > \O$, contradicting our assumption. To apply the former lemma effectively we need to break up the sums $\tilde{S}_i$ into subsums with $\lf_{\e_i}$-mass $\approx 1$; for that, we need to take care of $\tilde{X}_{i,j}$'s with big $\lf_{\e_i}$-mass first. 

 	 For each machine $i$, let $B_i$ be the set of indices $j$ such that $\lf_{\e_i}(\tilde{X}_{i,j}) > 1$ (``big items''). 
%
%
%
%
We need to show that the big items do not contribute much to \eqref{eq:mainCst}. 
		
		\begin{claim}
		 $\sum_i \frac{1}{v_i} \sum_{j \in B_i}  \lf_{\e_i}(\tilde{X}_{i,j}) \le 1.$ 
		\end{claim}
		
		\begin{proof}
			\renewcommand\qedsymbol{$\diamond$}
			First, from Corollary \ref{cor:simpleConv} we have that $\sum_{i,j} \E X_{i,j}^p \le (4\O)^p$, so passing to the tilde version and restricting to the big items we have
			\begin{align}
			\sum_{i,j \in B_i} \E \tilde{X}_{i,j}^p \le \frac{\O^p}{11^p}. \label{eq:bigItems}
			\end{align}
		Moreover, for the big items we can relate $\lf_{\e_i}(\tilde{X}_{i,j})$ and $\E \tilde{X}_{i,j}^p$. For that, again recall $(a+b)^p \le (2\max\{a,b\})^p \le 2^p(a^p + b^p)$; so expanding the definition of $\lf_{\e_i}$ we have for any random variable $X$
		\begin{align*}
			\lf_{\e_i}(X) \le \frac{1}{p} \ln\left[2^p \left(1 + \E \left(\frac{X}{\e_i}\right)^p \right) \right] \le 1 + \frac{1}{p} \ln\left[1 + \E \left(\frac{X}{\e_i}\right)^p \right] \le 1 + \frac{1}{p} \E \left(\frac{X}{\e_i}\right)^p = 1 + \frac{1}{p} \frac{1}{\e_i^p}\, \E X^p,
		\end{align*}
		where the last inequality uses that $\ln(1+x) \le x$ for all $x$. Moreover, for any big item we have by definition $\lf_{\e_i}(\tilde{X}_{i,j}) > 1$, so Lemma \ref{lemma:LBF} gives $\E (\tilde{X}_{i,j})^p \ge (\frac{\e_i}{10})^p$, or equivalently $1 \le \frac{10^p}{\e_i^p} \E (\tilde{X}_{i,j})^p$. Applying this bound and the displayed inequality to $\tilde{X}_{i,j}$, we can relate 	$\lf_{\e_i}(\tilde{X}_{i,j})$ to $\E \tilde{X}_{i,j}^p$: 
		\begin{align*}
			\lf_{\e_i}(\tilde{X}_{i,j}) \le \left(\frac{10^p}{\e_i^p} + \frac{1}{p} \frac{1}{\e_i^p}\right)\cdot \E (\tilde{X}_{i,j})^p \le \frac{11^p}{\e_i^p}\cdot \E (\tilde{X}_{i,j})^p = \frac{11^p v_i}{\O^p}\cdot\E (\tilde{X}_{i,j})^p.
		\end{align*}
		Dividing by $v_i$, adding this inequality over all big items, and employing \eqref{eq:bigItems} then concludes the proof. 
	\end{proof}
	
	\medskip
		Now assume by contradiction assume that \eqref{eq:mainCst} does not hold. Given this, and using the previous claim, if we remove the big items $\bigcup_i B_i$ from consideration we still have $\sum_i \frac{1}{v_i} (\sum_{j \notin B_i} \lf_{\e_i}(\tilde{X}_{i,j}) - 1) > 2$. So ignore the big items; to simplify the notation, we just assume there are no big items. Since $\lf_{\e_i}(\tilde{X}_{i,j}) \le 1$ for the remaining items, we can partition the sum $\tilde{S}_i = \sum_j \tilde{X}_{i,j}$ into subsums $\tilde{S}^0_i,\tilde{S}^1_i, \ldots, \tilde{S}^{k_i}_i$ such that $\tilde{S}^w_i$ has $\lf_{\e_i}$-mass in $[1,2]$ for all $w \ge 1$ and the exceptional sum $\tilde{S}^0_i$ has $\lf_{\e_i}$-mass at most $1$; formally we consider a partition $J_0,J_1,\ldots,J_{k_i}$ of the index set of $\{\tilde{X}_{i,j}\}_j$ such that $\tilde{S}^w_i := \sum_{j \in J_w} \tilde{X}_{i,j}$ has $\tilde{\lf}_{\e_i}(\tilde{S}^w_i) := \sum_{j \in J_w} \lf_{\e_i}(\tilde{X}_{i,j}) \in [1,2]$ for all $w \ge 1$, and $\tilde{\lf}_{\e_i}(\tilde{S}^0_i) \le 1$.
		
	Again $\|\tilde{S}\|_p$ can be lower bounded by ignoring the exceptional sums $\{\tilde{S}^0_i\}_i$ and assigning each of the other sums to their own coordinate, so 
	\begin{align}
		\E \|\tilde{S}\|_p \ge \E \bigg( \sum_{i,w \ge 1} (\tilde{S}^w_i)^p  \bigg)^{1/p}.  \label{eq:subload1}
	\end{align}
	We now lower bound the right-hand side using Theorem \ref{lemma:convJensen2}. First, using Lemma \ref{lemma:LBF} we have $\E (\tilde{S}^w_i)^p \ge (\frac{\e_i}{10})^p = \frac{\O^p}{10^p}\cdot \frac{1}{v_i}$. 
By scaling the $\tilde{S}^w_i$'s down if necessary, assume this holds at equality: $\E (\tilde{S}^w_i)^p = \frac{\O^p}{10^p}\cdot \frac{1}{v_i}$. Adding over all $i,w$,  
		\begin{align}
			\sum_{i,w\ge 0} \E (\tilde{S}^w_i)^p = \frac{\O^p}{10^p} \sum_i \frac{k_i}{v_i}. \label{eq:subload2}
		\end{align}
		In addition, the mass discounting the exceptional sums is at least 2:
	\begin{align*}
		\sum_i \frac{1}{v_i} \sum_{w = 1}^{k_i} \tilde{\lf}_{\e_i}(\tilde{S}^w_i) \ge \sum_i \frac{1}{v_i} \bigg(\sum_j \lf_{\e_i}(\tilde{X}_{i,j}) - 1\bigg) > 2.
	\end{align*}
	Since the $\tilde{\lf}_{\e_i}$'s in the left-hand side are at most 2, this implies that $\sum_i \frac{k_i}{v_i} > 1$. So applying this to \eqref{eq:subload2} we get $$\sum_{i,w\ge 0} \E (\tilde{S}^w_i)^p \ge \frac{\O^p}{10^p}.$$ Furthermore, since we assumed $v_i \ge (1/\alpha)^p$, we have $$\E (\tilde{S}^w_i)^p =  \frac{\O^p}{10^p}\cdot \frac{1}{v_i} \le \left(\frac{\alpha \O}{10}\right)^p.$$ But then applying Lemma \ref{lemma:convJensen2} to the $\frac{10}{\alpha \O} \tilde{S}^w_i$'s we get 
	\begin{align*}
		\E \bigg( \sum_{i,w \ge 1} (\tilde{S}^w_i)^p  \bigg)^{1/p} > \frac{1}{4} \bigg(\sum_{i,w \ge 1} \E (\tilde{S}^w_i)^p \bigg)^{1/p} \ge \frac{1}{4} \cdot \frac{\O}{10}. 
	\end{align*}
	Using \eqref{eq:subload1} and recalling that $S = 44 \cdot \tilde{S}$, we get $\E \|S\|_p > \O$, which contradicts the assumption $\E \|S\|_p \le \O$. This concludes the proof. 
	\end{proof}

	
	\paragraph{Converse bound.} Crucially, we need a converse to the previous theorem: if inequalities \eqref{eq:mainCst} are satisfied, then the $\ell_p$-norm of the loads is at most $O(\O)$. Indeed, one can show the following (with the additional control of the $\ell_\infty$-norm). 

	\begin{thm} \label{thm:mainConverse}
			 Consider a scalar $\O$ and a sufficiently small constant $\alpha$. Consider independent random variables $\{X_{i,j}\}_{i,j}$ in $[0,\alpha \O]$, and let $S_i = \sum_j X_{i,j}$. Suppose these random variables satisfy \eqref{eq:mainCst} for all sequences $v_1,\ldots,v_m \ge (1/\alpha)^p$. Also assume $\E \|S\|_\infty \le O(\O)$. Then $$\E \|S\|_p \le O(\O).$$
	\end{thm}
	
	We sketch a proof under simplifying assumptions (in which case we do not even need the condition $\E \|S\|_\infty \le O(\O)$); while we will actually require a modified version of this theorem, the simplified proof is helpful to provide intuition. 
	
	\begin{proof}[Proof idea of Theorem \ref{thm:mainConverse}]
		\emph{Assume} the following slightly stronger version of \eqref{eq:mainCst} holds for \emph{all} sequences $(v_i)_i$: $\sum_i \frac{1}{v_i} \sum_j \lf_{\O/v_i^{1/p}} \big(\tilde{X}_{i,j}\big) \le 3$. Applying this to the sequence $(\bar{v}_i)_i$ where $\bar{v}_i$ is such that $\sum_j \lf_{\O/\bar{v}_i^{1/p}} \big(\tilde{X}_{i,j}\big) \approx 1$, we get $\sum_i \frac{1}{\bar{v}_i} \lesssim 3$. , By Theorem \ref{thm:lFunc} $\frac{\O^p}{\bar{v}_i} \approx \E \tilde{S}_i^p$, and so we get $\sum_i \E \tilde{S}_i^p \lesssim 3 \O^p$, and inequality \eqref{eq:jensen} then gives $\E \|S\|_p \le O(\O)$, concluding the proof.
	\end{proof}	
	
	The issue with this theorem is that it will be hard to satisfy inequality \eqref{eq:mainCst} for all the allowed sequences $(v_i)_i$ later when we round our Linear Program. However, note that in the proof of this theorem we only needed this inequality to hold for a \emph{single} sequence $(\bar{v}_i)_i$ with specific properties, which will be easier to achieve. We will abstract out the properties needed. Actually, for technical reasons (controlling the size of the coefficients in the rounding phase of our algorithm) we will need to work with a capped version of $\lf$: $$\lfT_\e(X) := \min\{1, \lf_\e(X)\}.$$ In order to offset the loss introduced by this capping, we will also need a ``coarse control'' of the random variables (the result below holds without this coarse control if one uses $\lf$ instead of $\lfT$). Following~\cite{anupamSLB}, we will also use the effective size \eqref{eq:effectiveSize} to control the $\ell_\infty$-norm. This is then our main converse bound, whose proof is deferred to Appendix \ref{app:roundNeed}.
		
	\begin{thm} \label{thm:roundNeed}
		Consider a scalar $\O$ and a sufficiently small constant $\alpha$. Consider independent random variables $\{X_{i,j}\}_{i,j}$ in $[0,\alpha \O]$, and let $S_i = \sum_j X_{i,j}$, and $\tilde{X}_{i,j} = \frac{X_{i,j}}{44}$. Suppose the following hold:
		\begin{enumerate}
			\item ($\ell_p$ control) There exists an integer sequence $\bar{v}_1,\ldots,\bar{v}_m \ge (1/\alpha)^p$ such that for each $i$, either \linebreak$\sum_j \lfT_{\O/\bar{v}_i^{1/p}}(\tilde{X}_{i,j}) \le 10$ or $v_i = (1/\alpha)^p$, and $\sum_i \frac{1}{\bar{v}_i} \le 5$. \vspace{-3pt}
			
			\item (coarse control) $\sum_{i,j}\E  X_{i,j}^p \le O(\O)^p$.\vspace{-3pt}
			
			\item ($\ell_\infty$ control) There exists a sequence $\bar{\ell}_1,\ldots,\bar{\ell}_m \in [m]$ such that $\sum_j \beta_{\bar{\ell}_i}(X_{i,j}/\O) \le \gamma$ for all $i$, for some constant $\gamma$, and for each $\ell \in [m]$ at most $\ell$ of the $i$'s have $\bar{\ell}_i = \ell$.
		\end{enumerate}	
		Then $\E\|S\|_p \le O(\O)$.
	\end{thm}
	
	
	\section{Stochastic $\ell_p$ Load Balancing: Algorithm and analysis}
	
	In this section we prove Theorem \ref{thm:LB}, namely we give a constant-factor approximation to problem \stocps (please recall the definition of \stocps from Section \ref{sec:intro}). 
	
	Let $\OPT$ denote the smallest expected $\ell_p$ load \eqref{eq:obj} over all assignments of jobs to machines. The development of the algorithm mirrors that of the previous section and proceeds in 3 steps:
		
		\begin{enumerate}
			\item First we write an LP that essentially captures constraints \eqref{eq:mainCst} in a fractional way, which from Theorem~\ref{thm:mainCstr} we know to hold (after some truncation) for the optimal assignment (we also include a 		control on the $\ell_\infty$-norm using exponentially many constraints, as well the coarse control guaranteed by Corollary~\ref{cor:simpleConv}).
			
			\item Then, based on a fractional solution $\bar{x}$ of this LP, we write a reduced LP that is feasible (a crucial point) and imposes the requirements of Theorem \ref{thm:roundNeed} in a fractional way. This reduces the exponentially many inequalities for $\ell_p$ (and $\ell_\infty$) control to just one inequality per machine, by selecting the right $\bar{v}_i$'s (and $\bar{\ell}_i$'s) based on $\bar{x}$; this is done using the ideas in the proof sketch of Theorem~\ref{thm:mainConverse}.
			
			\item Since this reduced LP is much more structured and has fewer constraints, we can use an algorithm for the \GAP to find an integer approximate solution. Then from Theorem \ref{thm:roundNeed} the corresponding assignment has expected $\ell_p$-load $O(\OPT)$. 
		\end{enumerate}
		
		We make some simplifying assumptions. We consider the case $p \in (1,\infty)$, since the case $p=1$ is trivial (just assign assign job $j$ to the machine $i$ that gives the smallest expected job size $\E \J_{ij}$) and the case $p=\infty$ was solved in~\cite{anupamSLB}. By using a standard binary search argument we assume throughout that we have an estimate of the optimal value $\OPT$ within a factor of 2 (i.e., if our starting LP is feasible we reduce the current estimate of $\OPT$, and if it is infeasible we increase it). In fact, to simplify the notation we assume we know $\OPT$ exactly: the error in the estimation translates directly to the constants in the approximation factor. 
			

	\subsection{Starting LP} \label{sec:startingLP}
	
 As in~\cite{KRT,anupamSLB}, we split the job into its \emph{truncated} and \emph{exceptional} parts: Let $\alpha$ be a sufficiently small constant (with $1/\alpha$ integral, to simplify things); we then define the truncated part $\tJ_{ij} = \J_{ij} \cdot \one(\J_{ij} \le \alpha\OPT)$, and the exceptional part $\rJ_{ij} = \J_{ij} \cdot \one(\J_{ij} > \alpha\OPT)$, where $\one(E)$ is the indicator of the event $E$ (notice $\J_{ij} = \tJ_{ij} + \rJ_{ij}$).
 		
		Our LP, with variable $x_{ij}$ denoting the amount of job $j$ assigned to machine $i$, is then the following (as before we use tildes to denote the scaling $\tildeJ_{ij} := \nicefrac{\tJ_{ij}}{44}$):
	\begin{align}
		&\sum_{i,j} (\E \rJ_{ij})\,x_{ij} \le 2 \OPT \label{eq:constrLP0}\\
		&\sum_i \frac{1}{v_i} \bigg(\sum_j  \lfT_{\OPT/v_i^{1/p}}(\tildeJ_{ij})\,x_{ij} - 1\bigg) \le 3~~~~~\forall v_i \in \{1/\alpha^p,\ldots,m\},~~ \forall i \in [m], \label{eq:constrLP1}\\
		&\sum_{i \in K} \sum_j \beta_k(\tJ_{ij}/\OPT)\,x_{ij} \le C \cdot k ~~~~~~~~~~~~~~~\forall K \subseteq [m] \textrm{ with $|K| = k$},~\forall k \in [m]  \label{eq:constrLP2}\\
		&\sum_{i,j} (\E \J_{ij}^p) x_{ij} \le (4\OPT)^p \label{eq:constrLPnew}\\
		&x \in \textrm{assignment polytope}, \label{eq:constrLP3}
	\end{align}
	where $C$ is a sufficiently large constant, and the assignment polytope is the standard one $\{x \in [0,1]^{n \times m} : \sum_i x_{ij} = 1~ \forall j\}$. 	Constraint \eqref{eq:constrLP0} is borrowed from~\cite{anupamSLB} and controls the contribution to the $\ell_p$-norm by the exceptional parts. Constraints \eqref{eq:constrLP1} capture a weakened version of the bounds guaranteed by Theorem \ref{thm:mainCstr} (notice $\lfT \le \lf$); as mentioned earlier, what we gain from this weakening is a better control on the size of the coefficients, important for the rounding step.  Constraint \eqref{eq:constrLP2} is also from~\cite{anupamSLB} and controls the $\ell_\infty$-norm of the truncated part. Constraint \eqref{eq:constrLPnew} imposes the bound guaranteed by Corollary \ref{cor:simpleConv} and is only required to control the loss incurred by using the capped quantity $\lfT$ instead of~$\lf$.

	Lemma 2.3 of~\cite{anupamSLB} shows that the optimal (integral) assignment satisfies constraints \eqref{eq:constrLP0} and \eqref{eq:constrLP2} (notice that since $\|.\|_{\infty} \le \|.\|_p$, the loads of the optimal solution satisfies $\E \|S\|_\infty \le \OPT$). Applying Theorem \ref{thm:mainCstr} and Corollary \ref{cor:simpleConv} with $\{X_{ij}\}_j$ representing the truncated part of the items assigned to machine $i$ by this solution and with $\O = \OPT$, we see that constraints \eqref{eq:constrLP1} and \eqref{eq:constrLPnew} are also satisfied by the optimal solution. Therefore, the LP is feasible. 	
	
	About solving it in polynomial time: Notice that we can write the inequalities \eqref{eq:constrLP1} by setting an auxiliary variable $z_i$ with $$z_i \ge \frac{1}{v_i} \bigg(\sum_j  \lf_{\OPT/v_i^{1/p}}(\tildeJ_{ij})\,x_{ij} - 1\bigg)~~~\forall v_i \in \{1/\alpha^p,\ldots,m\},$$ and replacing constraint \eqref{eq:constrLP1} by just $\sum_i z_i \le 3$. Thus, we can capture all constraints except \eqref{eq:constrLP2} with a poly-sized formulation. Since it is easy to see that we can separate inequalities \eqref{eq:constrLP2} in poly-time (see Section 2.3 of \cite{anupamSLB}), we can use the ellipsoid method to solve the LP in polynomial time. Summarizing this discussion we have the following.
	
	\begin{prop}
		The LP \eqref{eq:constrLP0}-\eqref{eq:constrLP3} is feasible, and can be solved in polynomial time. 
	\end{prop}
	
	
	\subsection{The reduced LP}

	So suppose we have a feasible fractional solution $\bar{x}$ for the LP \eqref{eq:constrLP0}-\eqref{eq:constrLP3}. It seems we cannot hope to round it to an integral solution and satisfy all the constraints with reasonable loss. However, to control the $\ell_p$-norm of the truncated parts we only need the integral assignment to satisfy the requirements of Theorem \ref{thm:roundNeed} (the exceptional parts will not be problematic). So we will simplify the LP \eqref{eq:constrLP0}-\eqref{eq:constrLP3} by selecting for each machine a \emph{single} $\bar{v}_i$ (based on the proof of Theorem \ref{thm:mainConverse}) and $\bar{\ell}_i$ (based on a simplification of~\cite{anupamSLB}) as follows:
	
	\begin{enumerate}
		\item Set $\bar{v}_i$ to be the largest value in $\{1/\alpha^p, \ldots, m\}$ such that $\sum_j \lfT_{\OPT/\bar{v}_i^{1/p}}(\tildeJ_{ij})\, \bar{x}_{ij} \le 2$, if it exists. Notice that $\sum_j \lfT_{\OPT/v^{1/p}}(\tildeJ_{ij})\, \bar{x}_{ij}$ is a continuous increasing function of $v$, so if the desired $\bar{v}_i$ does not exist it means that $\sum_j \lfT_{\alpha \OPT}(\tildeJ_{ij})\, \bar{x}_{ij} > 2$. Let $I \subseteq [m]$ be the set of machines for which the desired $\bar{v}_i$ exists. For all $i \notin I$, we set $\bar{v}_i = 1/\alpha^p$. 
		
		\item Let $\bar{\ell}_i \in [m]$ be the largest such that $\sum_j \beta_{\bar{\ell}_i}(\tJ_{ij}/\OPT)\,\bar{x}_{ij} \le C$, where $C$ is the constant in constraints \eqref{eq:constrLP2} (such $\bar{\ell}_i$ exists since constraint \eqref{eq:constrLP2} implies that setting it to 1 satisfies this inequality).
	\end{enumerate}
	
	The reduced LP then becomes:
	\begin{align}
		&\sum_{i,j} (\E \rJ_{ij})\,x_{ij} \le 2 \OPT \label{eq:constrIP0}\\	
		&\sum_j \lfT_{\OPT/\bar{v}_i^{1/p}}(\tildeJ_{ij})\, x_{ij} \le 2~~~~~~~~~~~ \forall i \in I \label{eq:constrIP1}\\
		&\sum_j \beta_{\bar{\ell}_i}(\tJ_{ij}/\OPT)\,x_{ij} \le 1~~~~~~~~~~~~~\forall i \label{eq:constrIP2}\\
		&\sum_{i,j} (\E \J_{ij}^p) x_{ij} \le (4\OPT)^p \label{eq:constrIPnew}\\
		&x \in \textrm{assignment polytope}. \label{eq:constrIP3}
	\end{align}
	Notice that by construction $\bar{x}$ is a fractional solution to this LP, so in particular the LP is feasible.
	
	Now we analyze the quality of an integral assignment satisfying approximately this LP; we will see how to obtain such integral assignment in the next section. First, we start by remarking that \eqref{eq:constrIP0} fully controls the exceptional parts of the jobs.
	
	\begin{lemma}
	Consider an integral assignment  $x \in \{0,1\}^{n \times m}$ satisfying constraint \eqref{eq:constrIP0} within a multiplicative factor (i.e., with RHS replaced by $O(\OPT)$). Let $S''_i = \sum_j \rJ_{ij} x_{ij}$ be the load incurred on machine $i$ by the exceptional sizes of jobs assigned to it. Then $\E\|S''\|_p \le O(\OPT)$. 
	\end{lemma}
	
	\begin{proof}
		Since the $\ell_p$-norm is always at most the $\ell_1$-norm and the $\rJ_{ij}$'s are non-negative, we have 
		\begin{align*}
			\E\|S''\|_p \le \E\|S''\|_1 = \sum_{ij} (\E \rJ_{ij})x_{ij} \le O(\OPT). \tag*{\qedhere}
		\end{align*} 
	\end{proof}
	
	In addition, any integral assignment approximately satisfying constraints \eqref{eq:constrIP1}-\eqref{eq:constrIPnew} fulfills the requirement of Theorem \ref{thm:roundNeed} for the truncated parts, and thus we can control their expected $\ell_p$ norm. 

	\begin{lemma}
	Consider an integral assignment  $x \in \{0,1\}^{n \times m}$ satisfying constraints \eqref{eq:constrIP1}-\eqref{eq:constrIPnew} within a multiplicative factor of 5. For each $i$, let $\{X_{i,j}\}_j = \{\tJ_{ij} x_{ij}\}_j$ (i.e., the truncated part of the jobs assigned to machine $i$). Then $\{X_{i,j}\}_{i,j}$, $\{\bar{v}_i\}_i$, and $\{\bar{\ell}_i\}_i$ satisfy the requirements of Theorem \ref{thm:roundNeed} with $\O = \OPT$. 
	
	In particular, letting $S'_i = \sum_j \tJ_{ij}$ be the load incurred on machine $i$ by the truncated sizes of jobs assigned to it, we have $\E\|S'\|_p \le O(\OPT)$. 
	\end{lemma}
	
	\begin{proof}
		The second part of the lemma follows directly from Theorem \ref{thm:roundNeed}, so we prove the first part. 		First, from the definition of the truncation we have $X_{i,j} \le \alpha \OPT$. We show that Item 1 ($\ell_p$ control) in Theorem~\ref{thm:roundNeed} holds; since $x$ satisfies constraints \eqref{eq:constrIP1} within a multiplicative factor of 5, and by the choice of the $\bar{v}_i$'s, it suffices to show $\sum_i \frac{1}{\bar{v}_i} \le 5$. 
		
		We partition the indices $i$ into 2 sets, depending on whether $\bar{v}_i$ hit the upper bound $m$ or not: $U_{<m} = \{i \in [m] : \bar{v}_i < m\}$ and $U_m = \{i \in [m] : \bar{v}_i = m\}$. By definition $\sum_{i \in U_m} \frac{1}{\bar{v}_i} \le m \cdot \frac{1}{m} = 1$. For an index $i \in U_{< m}$, by maximality of $\bar{v}_i$ we have that $\bar{v}_i + 1$ satisfies
		\begin{align*}
		V_i := \sum_j \lf_{\OPT/(\bar{v}_i + 1)^{1/p}}(\tildeJ_{ij})\, \bar{x}_{ij} > 2,
		\end{align*}
		and hence $V_i - 1\ge 1$. But since $\bar{x}$ satisfies constraints \eqref{eq:constrLP1}, we have $\sum_{i \in U_{<m}} \frac{1}{\bar{v}_i + 1} (V_i - 1) \le 3$, and hence $\sum_{i \in U_{<m}} \frac{1}{\bar{v}_i + 1} \le 3$. Finally, since $\bar{v}_i \ge (1/\alpha)^p \ge 100$ (since $\alpha$ is a sufficiently small constant), we have $\frac{1}{\bar{v}_i} \le 1.01 \frac{1}{\bar{v}_i + 1}$ and hence $\sum_{i \in U_{<m}} \frac{1}{\bar{v}_i} \le 4$. This shows $\sum_i \frac{1}{\bar{v}_i} \le 5$. 
		
		Item 2 (coarse control) in Theorem~\ref{thm:roundNeed} is directly enforced by constraint \eqref{eq:constrIPnew}. To show that Item 3 ($\ell_\infty$ control) in Theorem~\ref{thm:roundNeed} holds, we just need that for all $\ell \in [m]$, for at most $\ell$ of the $i$'s we have $\bar{\ell}_i = \ell$. Since this is clearly true for $\ell = m$, consider $\ell < m$ and suppose by contradiction that there there is a set $K \subseteq [m]$ of size $\ell + 1$ such that $\bar{\ell}_i = \ell$ for all $i \in K$. By maximality of $\bar{\ell}_i$, for all $i \in K$ we have $\sum_j \beta_{\ell + 1}(\tJ_{ij}/\OPT)\,\bar{x}_{ij} > C$; adding this over all $i \in K$ and using that $\bar{x}$ satisfies constraint \eqref{eq:constrLP2} for $K$, we have
		\begin{align*}
			C \cdot (\ell + 1) < \sum_{i \in K} \sum_j \beta_{\ell + 1}(\tJ_{ij}/\OPT)\,\bar{x}_{ij} \stackrel{\eqref{eq:constrLP2}}{\le} C \cdot (\ell + 1),
		\end{align*}
		reaching a contradiction. This concludes the proof. 
	\end{proof}
	
	Since the total size of a job equals its truncated plus its exceptional part, the previous lemmas and triangle inequality give the following. 
	
	\begin{cor} \label{cor:mainGuarantee}
		Consider an integral assignment $x \in \{0,1\}^{n \times m}$ satisfying constraints \eqref{eq:constrIP0}-\eqref{eq:constrIPnew} within a factor of 5. Then letting $S_i = \sum_j \J_{ij} x_{ij}$ be the load incurred on machine $i$, we have $\E \|S\|_p \le O(\OPT)$. 
	\end{cor}

	
\subsection{Finding an approximate integral solution to the reduced LP} 

	The main observation is that the LP \eqref{eq:constrIP0}-\eqref{eq:constrIP3} is essentially that of the \GAP (\GAPs). In (the feasibility version of) this problem, we again have $m$ machines and $n$ jobs, a precessing time $a_{ij}$ and cost $b_{ij}$ for assigning job $j$ to machine $i$. Given budgets budgets $\{A_i\}_i$ and $B$ respectively, the goal is to find an integral solution $x$ to the system 
	\begin{align}
		&\sum_j a_{ij} x_{ij} \le A_i ~~~~~~~\forall i \in [m] \label{eq:cstMake}\\
		&\sum_{i,j} b_{ij} x_{ij} \le B \label{eq:cstCost}\\
		&x \in \textrm{ assignment polytope}.
	\end{align}
	Shmoys and Tardos~\cite{ST} designed an algorithm that given any fractional solution to the above program produces an integral assignment that satisfies \eqref{eq:cstCost} exactly, and satisfies constraints \eqref{eq:cstMake} with the RHSs increased to $A_i + \max_j a_{ij}$.
	
Notice that the reduced LP \eqref{eq:constrIP0}-\eqref{eq:constrIP3} is essentially an instance of \GAPs: the difference is that we have 2 cost-type constraints and 2 makespan-type constraints for some machines. But we can simply combine the equations of the same type to obtain a \GAPs instance: add $\frac{1}{2\OPT}$ of inequality \eqref{eq:constrIP0} to $\frac{1}{(4\OPT)^p}$ of inequality \eqref{eq:constrIPnew} to form a single cost constraint with RHS 2, and add $\frac{1}{2}$ of inequality \eqref{eq:constrIP1} to \eqref{eq:constrIP2} for each $i \in I$ to obtain a single makespan constraint constraint with RHS 2 (for $i \notin I$ just keep the makespan constraint \eqref{eq:constrIP2}).

	Since this \GAPs instance is a relaxation of the LP \eqref{eq:constrIP0}-\eqref{eq:constrIP3} it is also feasible. Thus, consider any fractional solution to this \GAPs instance and let $\tilde{x}$ be the integral assignment produced by Shmoys-Tardos algorithm~\cite{ST}.
	
	\begin{lemma} \label{lemma:mainAlgo}
		The integral assignment $\tilde{x}$ satisfies all the constraints \eqref{eq:constrIP0}-\eqref{eq:constrIPnew} within a factor of 2. 
	\end{lemma} 
	
	\begin{proof}	
	Notice that for this \GAPs instance $A_i \ge 1$ and $$\max_j a_{ij} \le 		\max\left\{\lfT_{\OPT/\bar{v}_i^{1/p}}(\tildeJ_{ij})\,,\,\beta_{\bar{\ell}_i}(\tJ_{ij}/\OPT)\right\};$$ by construction $\lfT_\e \le 1$ (this is the only motivation for introducing this capped version of $\lf$) and since the truncated sizes have $\tJ_{ij} \le \OPT$
		\begin{align*}
		\beta_{\bar{\ell}_i}(\tJ_{ij}/\OPT) = 	\frac{1}{\ln \bar{\ell}_i}\, \ln \E \exp\left(\frac{\tJ_{ij}}{\OPT} \ln \bar{\ell}_i\right) \le \frac{1}{\ln \bar{\ell}_i}\, \ln e^{\ln \bar{\ell}_i} = 1,
		\end{align*}	
	and so $\max_j a_{ij} \le 1$ and hence $A_i + \max_j a_{ij} \le 2 A_i$ for all $i$. Thus, by the guarantees of~\cite{ST} $\tilde{x}$ satisfies constraints \eqref{eq:cstMake} and \eqref{eq:cstCost} within a multiplicative factor of 2. The fact that all the coefficients are non-negative then implies that $\tilde{x}$ satisfies the disaggregated constraints within a multiplicative factor of 4 (i.e., apply that for non-negative $u_{ij}$'s and $v_{ij}$'s, $\sum_{ij} (u_{ij} + v_{ij}) \tilde{x}_{ij} \le 2 \cdot 2$ implies $\sum_{ij} u_{ij} \le 4$, and the same for the $v_{ij}$'s). This concludes the proof. 
	\end{proof}
	
	Then from Corollary \ref{cor:mainGuarantee} the assignment $\tilde{x}$ has expected $\ell_p$ load at most $O(\OPT)$. This proves Theorem~\ref{thm:LB}.
	


	\bibliographystyle{alpha}
	\bibliography{online-lp-short}	


	\pagebreak
	\appendix

		\section{Proof of Theorem \ref{thm:pSel}} \label{app:pSel}

 Consider an instance of $\pSel$, and suppose we have an $\alpha$-approximate linear optimization oracle over $\P$, for a constant $\alpha$. Let $x^*$ be an optimal solution for $\pSel$, let $S = \sum_{j : x^*_j = 1} V_j$, and let $\OPT = (\E S^p)^{1/p}$ be the value of this optimal solution. By a standard binary search argument, assume $\OPT$ is known within a constant factor (i.e., if \eqref{eq:sel} is feasible we increase the estimate of $\OPT$, if it is infeasible we decrease it). In fact, to simplify the notation we assume we know $\OPT$ exactly: the error in the estimation translates directly to the constants in the approximation factor. 
	
 Define $\bar{\e} = \frac{\OPT}{e^{1/\alpha}}$. From Lemma \ref{lemma:UBF} we have that $\sum_{j : x^*_j = 1} \lf_{\bar{\e}}(V_j) \ge \frac{1}{\alpha}$.  Therefore, the optimal solution $x^*$ is feasible for the program
	\begin{align}
		&\sum_j \lf_{\bar{\e}}(V_j)\, x_j \ge \frac{1}{\alpha} \label{eq:sel}\\
		&x \in \P. \notag
	\end{align}
	Then use an $\alpha$-approximate linear optimization oracle over $\P$ to find a solution $\bar{x} \in \P$ satisfying $\sum_j \lf_{\bar{\e}}(V_j) \bar{x}_j \ge 1$ (if cannot find, increase the estimate of $\OPT$). Lemma \ref{lemma:LBF} then implies that $\E(\sum_j V_j \bar{x}_j)^p \ge (\frac{\bar{\e}}{10})^p = \OPT^p\,\frac{1}{(10 e^{1/\alpha})^p}$; since $\alpha$ is a constant, this implies that $(\E(\sum_j V_j \bar{x}_j)^p)^{1/p} \ge \Omega(\OPT)$ and concludes the proof of Theorem \ref{thm:pSel}.
 	
	
	\section{Proof of Corollary \ref{cor:simpleConv}} \label{app:simpleConv}
	
	We prove the contrapositive: assume $\sum_{i,j} \E X_{i,j}^p > (4\O)^p$; we want to show $\E \|S\|_p > \O$. 	Notice the total load $\|S\|_p$ is at least the load of putting each of the $X_{i,j}$'s in their own coordinate, namely
		\begin{align*}
			\|S\|_p \ge \|(X_{i,j})_{i,j}\|_p.
		\end{align*}
		Applying Lemma \ref{lemma:convJensen2} to the scaled right-hand side vector $\frac{1}{\alpha \O} (X_{i,j})_{i,j}$ we obtain that $$\E \|(X_{i,j})\|_p \ge \frac{1}{4} \bigg(\sum_{i,j} \E X_{i,j}^p\bigg)^{1/p} \ge \O;$$ notice we can indeed apply this lemma since $\frac{X_{i,j}}{\alpha \O}  \in [0,1]$ and hence $\E\big(\frac{X_{i,j}}{\alpha \O}\big)^p \le 1$, but $\sum_{i,j} \E \big(\frac{X_{i,j}}{\alpha \O}\big)^p \ge (\frac{1}{\alpha})^p$. Putting the displayed inequalities together concludes the proof.

	\section{Proof of Theorem \ref{thm:roundNeed}} \label{app:roundNeed}

	Let $$I := \bigg\{ i : \sum_j \lfT_{\O/\bar{v}_i^{1/p}}(\tilde{X}_{i,j}) \le 10\bigg\}$$ and let $I^c = [m] \setminus I$ be its complement. Also, for each coordinate $i$, let $J_i = \Big\{j : \lfT_{\O/\bar{v}_i^{1/p}}(\tilde{X}_{i,j}) = \lf_{\O/\bar{v}_i^{1/p}}(\tilde{X}_{i,j})\Big\}$ be the set of $j$'s where the capping of $\lf$ did not made a difference, and let $J^c_i$ be its complement. Let $S^J$ be the vector with coordinates $S^J_i = \sum_{j \in J_i} X_{i,j}$ (so only contributions from $j$'s in $J_i$), and let $S^{J^c} = S - S^J$ be the other contributions. We will break up $S$ as $S = S^J_I + S^{J^c}_I + S_{I^c}$. From triangle inequality is suffices to show that the expected $\ell_p$-norm of each term is at most $O(\O)$. 
	
	We start with $S_I^J$. By definition of $I$ and $J$, for each $i \in I$ we have $$\sum_{j \in J_i} \lf_{\O/\bar{v}_i^{1/p}}(\tilde{X}_{i,j}) = \sum_{j \in J_i} \lfT_{\O/\bar{v}_i^{1/p}}(\tilde{X}_{i,j}) \le 10,$$ and so Lemma \ref{lemma:UBF} gives that $\E (\tilde{S}^J_i)^p \le \frac{\O^p}{\bar{v}_i} e^{10p}$, where $\tilde{S}^J_i := \frac{S^J_i}{44}$ as usual. Moreover, since the sequence satisfies $\sum_i \frac{1}{\bar{v}_i} \le 5$ and all terms are positive, in particular we have $\sum_{i \in I} \frac{1}{\bar{v}_i} \le 5$; so adding over all $i \in I$ we have $$\sum_{i \in I} \E (\tilde{S}_i^J)^p \le \O^p e^{10p} \sum_{i \in I} \frac{1}{\bar{v}_i} \le 5 \O^p e^{10p}.$$ Then the Jensen's based inequality \eqref{eq:jensen} gives $\E \|\tilde{S}^J_{I}\|_p \le O(\O)$, and hence $\E \|S^J_{I}\|_p \le O(\O)$.

	Now we upper bound $\E \|S^{J^c}_I\|_p$. Since for each $j \in J^c_i$ the capping of $\lf$ kicked in, we have $\lfT_{\O/\bar{v}_i^{1/p}}(\tilde{X}_{i,j}) = 1$. Thus, by definition of $I$, for $i \in I$ there can be at most 10 elements in $J_i^c$. Employing the elementary inequality $$\left(\sum_{u \in U} u \right)^p \le \left(|U| \max_{u \in U} u\right)^p \le |U|^p \sum_{u \in U} u^p$$ which holds for any set $U$ of non-negative numbers, we obtain $\E (S^{J^c}_i)^p \le |J_i|^p \sum_{j \in J_i} \E X_{i,j}^p \le 10^p \sum_{j \in J_i} \E X_{i,j}^p $ for each $i \in I$. Adding over all $i \in I$ and using the ``coarse control'' Item 2 of the lemma, we have $\sum_{i \in I} \E(S^{J^c}_i)^p \le 10^p O(\O)^p$. Again inequality \eqref{eq:jensen} gives that $\E \|S^{J^c}_I\|_p \le O(\O)$. 

	Finally we control $S_{I^c}$. First, there are not too many coordinates in $I^c$: again we have $\sum_{i \in I^c} \frac{1}{\bar{v}_i} \le 5$, and since for such $i$'s $\bar{v}_i = (1/\alpha)^p$, this implies that $|I^c| \le 5 \alpha^p$. Moreover, because of Item 3 of the lemma ($\ell_\infty$ control), the arguments from Lemma 2.4 of~\cite{anupamSLB} show that $\E\|S_{I^c}\|_{\infty} \le O(\O)$ (for completeness we provide a proof in Appendix \ref{sec:linfty}). Finally, from $\ell_p$-$\ell_\infty$ comparison (i.e., for any vector $x \in \R^d$, $\|x\|_p \le d^{1/p} \|x\|_{\infty}$ holds) we have $\|S_{I^c}\|_p \le |I^c|^{1/p} \|S_{I^c}\|_{\infty}$, and thus $\E \|S_{I^c}\|_p \le O(\O)$. This concludes the proof of the theorem.

	
	\section{Controlling the $\ell_\infty$-norm} \label{sec:linfty}
	
	For completeness we prove the following lemma, following the proof of Lemma 2.8 of \cite{anupamSLB}.
	
	\begin{lemma} \label{lemma:linfty}
		Consider independent random variables $\{X_{i,j}\}_{i,j}$, and let $S_i = \sum_j X_{i,j}$. Suppose there is a sequence $\ell_1,\ldots,\ell_m \in [m]$ such that $\sum_j \beta_{\ell_i}(X_{i,j}/\O) \le \gamma$ for all $i$, for some constant $\gamma$, and for each $\ell \in [m]$, at most $\ell$ of the $i$'s have $\ell_i = \ell$. Then $\E \|S\|_\infty \le O(\O)$. 
	\end{lemma}
	
	To prove this we need the following lemma (Lemma 2.1 of \cite{anupamSLB}), which follows directly by the Chernoff-Cram\`er method for concentration. 
	
	\begin{lemma} \label{lemma:simpleConc}
		For any independent random variables $Y_1,\ldots,Y_n$, we have $$\Pr\left(\sum_j Y_j \ge \sum_j \beta_\ell(Y_j) + t \right) \le \ell^{-t} ~~~~~\forall t \ge 0.$$
	\end{lemma}
	
	\begin{proof}[Proof of Lemma \ref{lemma:linfty}]
		Let $I = \{i : \ell_i \le 2\}$, and notice that by assumption $|I| \le 3$. Jensen's inequality shows that for any random variable $X$ and any $\ell$, $\beta_\ell(X) \ge \E X$. Therefore, our assumption implies 
		\begin{align*}
			\E \|S_I\|_{\infty} \le \sum_{i \in I} \E S_i \le \O \cdot \sum_{i \in I} \sum_j \beta_{\ell_i}(X_{i,j}/\O) \le 2\gamma. 
		\end{align*}
		To upper bound $\E\|S_{I^c}\|_{\infty}$, where $I^c = [m] \setminus I$, we apply Lemma \ref{lemma:simpleConc} to $\frac{S_i}{\O} = \frac{\sum_j X_{i,j}}{\O}$ to obtain for all $t \ge 3$
		\begin{align*}
			\Pr\left(\frac{S_i}{\O} \ge \sum_j \beta_{\ell_i}(X_{i,j}/\O) + t \right) \le \ell_i^{-t}
		\end{align*}
		and then a union bound over all $i \in I^c$ to obtain
		\begin{align*}
			\Pr(\|S_{I^c}\|_{\infty} \ge \O (\gamma + t)) &\le \sum_{i \in I^c} \Pr\left(\frac{S_i}{\O} \ge \sum_j \beta_{\ell_i}(X_{i,j}/\O) + t \right) \le \sum_{\ell = 3}^m \sum_{i : \ell_i = \ell} \ell_i^{-t}\\
			&\le \sum_{\ell = 3}^m \ell_i^{-t + 1} \le \int_2^{\infty} x^{-t + 1}\,\d x = \frac{2^{-t + 2}}{t-2} \le 2^{-t + 2},
		\end{align*}
		where the third inequality uses our assumption. Using the nonegativity of $\|S_{I^c}\|_{\infty}$, we integrate the tail to obtain
		\begin{align*}
			\E \|S_{I^c}\|_{\infty} &= \int_0^\infty \Pr(\|S_{I^c}\|_{\infty} \ge t)\, \d x \\
			&= \int_0^{\O \cdot (\gamma + 3)} \Pr(\|S_{I^c}\|_{\infty} \ge t)\, \d x + \O \cdot \int_3^{\infty} \Pr(\|S_{I^c}\|_{\infty} \ge \O\, \gamma + \O\, t)\, \d x\\
			&\le \O \cdot (\gamma + 3) + \int_3^{\infty} 2^{-t + 2} \,\d x \le O(\O).
		\end{align*}
		Since by triangle inequality $\|S\|_\infty \le \|S_I\|_{\infty} + \|S_{I^c}\|_{\infty}$, putting the above bounds together concludes the proof. 
	\end{proof}

\end{document}